\documentclass[final,1p,times,sort&compress]{elsarticle}
\journal{$\;$}

\usepackage{etoolbox}

\makeatletter
\patchcmd{\ps@pprintTitle}% <cmd>
  {Preprint submitted to}% <search>
  {Preprint}% <replace>
  {}{}% <succes><failure>
\makeatother

\usepackage[utf8]{inputenc}

\usepackage{subfigure}

\usepackage{enumerate}
\usepackage{amsmath}
\usepackage{hyperref}

\newcommand{\Mod}[1]{\ (\mathrm{mod}\ #1)}

\newtheorem{thm}{Theorem}
\newtheorem{theorem}[thm]{Theorem}
\newtheorem{lemma}[thm]{Lemma}
\newtheorem{corollary}[thm]{Corollary}

\newtheorem{proof}{Proof}

\makeatletter
\newcommand{\nosemic}{\renewcommand{\@endalgocfline}{\relax}}% Drop semi-colon ;
\newcommand{\dosemic}{\renewcommand{\@endalgocfline}{\algocf@endline}}% Reinstate semi-colon ;
\let\oldnl\nl% Store \nl in \oldnl
\newcommand{\nonl}{\renewcommand{\nl}{\let\nl\oldnl}}% Remove line number for one line
\makeatother
\DeclareRobustCommand{\&}{%
  \ifmmode\expandafter\mathbin\fi\char`&
}

\newcommand{\zQbb}{\mathbb{Q}}

\newcommand{\zPcal}{\mathcal{P}}

\newcommand{\zWcal}{\mathcal{W}}

\newcommand{\nc}{\newcommand}
\nc{\blacktoken}{\mbox{\sl black}}% Colorless token

\begin{document}
%\publicationdetails{VOL}{2020}{ISS}{NUM}{SUBM}
%\maketitle

%\input{parts/0_abstract.tex}

%=======================================================================
% Import illustrations
%=======================================================================

%\input{img/boxpicI}

%=======================================================================
% End of: Import illustrations
%=======================================================================

%\mainmatter  % start of an individual contribution

%ELS
\begin{frontmatter}

%\title{On Existence of Linear Dynamical Metric Graphs with
%Slowest Growth Rate
\title{Towards Dynamic-Point Systems on Metric Graphs with
Longest Stabilization Time
%\thanks{This work is supported by the Basic Research Program of
%the National Research University Higher School of Economics and
%Russian Foundation for Basic Research, project No. 16-37-00482 mol\_a.}
}
%\titlerunning{ }

%ELS
\author[nru]{Leonid W. Dworzanski}% \corref{cor1}}
%\author{ \\ \texttt{leo@mathtech.ru} } % \corref{cor1}}
\ead{leo@mathtech.ru}

%\author{Leonid W. Dworzanski\affiliationmark{}\thanks{Acknowledgments: The reported study was funded by RFBR, project number 20-07-01103}}
%\title[Towards DP-systems on metric graphs with longest stabilization time]{Towards Dynamic-Point Systems on Metric Graphs with
%Longest Stabilization Time}
% put your affiliation here, not your full address.
% If you like to give away your email or other parts of your address,
% THIS IS NOT THE RIGHT PLACE, your address will change, this paper
% will not.
% Just watch that your personal data that you want to communicate on
% the episcience server is always up to date.
%\affiliation{
%  % one line per affiliation, no postal codes, grant numbers or similar
%  National Research University Higher School of Economics, Russia%\\
%  %ICube, university of Strasbourg, France\\
%  %Alma Mater, campus universalis, terra incognita
%  }
%\keywords{quantum graphs, stabilization time, system of dynamic points,
%discrete event dynamic systems}
% %%don't try to cheat here, we will check the dates!
%\received{2020-10-12}
%\revised{xxxx-xx-xx}
%\accepted{xxxx-xx-xx}

%\maketitle

\address[nru]{National Research University Higher School of Economics,
Myasnitskaya ul. 20, Moscow, Russia 101000}


%\toctitle{Lecture Notes in Computer Science}
%\tocauthor{Authors' Instructions}

%%%%
%%%%
%%%%% first the title is needed
%%%%\title{On Compositionality of 1L-Liveness and Local Path-Based Properties
  %%%%for Nested Petri nets 
%%%%%\thanks{This work is supported by the Basic Research Program of
%%%%%the National Research University Higher School of Economics and
%%%%%Russian Foundation for Basic Research, project No. 16-37-00482 mol\_a.}
%%%%}
%%%%% a short form should be given in case it is too long for the running head
%%%%\titlerunning{ }
%%%%
%%%%\author{Leonid W. Dworzanski%
%%%%}%
%%%%\authorrunning{ }
%%%%%
%%%%\institute{National Research University Higher School of Economics,\\
%%%%Myasnitskaya ul. 20, 101000 Moscow, Russia\\
%%%%leo@mathtech.ru \\
%%%%%\url{http://www.hse.ru}
%%%%}
%%%%
%%%%\toctitle{Lecture Notes in Computer Science}
%%%%\tocauthor{Authors' Instructions}
%%%%

\begin{abstract}

A dynamical system of points moving along the edges of a graph could
be considered as a geometrical discrete dynamical system or as
a discrete version of a quantum graph with localized wave packets.
%that have longest stabilization time are studied.
We study the set of such systems over metric graphs that can be constructed
from a given set of commensurable edges with fixed lengths.
It is shown that there always exists a system consisting of a bead graph with
vertex degrees not greater than three that demonstrates the longest stabilization
time in such a set.
The results are extended to graphs with incommensurable edges
using the notion of $\varepsilon$-nets and,
also, it is shown that dynamical systems of points on linear
graphs %with one initial point
have the slowest growth of the number of dynamic points.
%
%structural properties of dynamic-point systems
%%over such metric graphs
%that demonstrate the longest stabilization time are considered.
%It is shown that,
%in such a set of dynamical systems with the longest stabilization time,
%there is a system consisting of a bead graph
%with vertex degrees not greater than three.
\end{abstract}

%ELS

%ELS
\end{frontmatter}

%ELS

\begin{abstract}

A dynamical system of points moving along the edges of a graph could
be considered as a geometrical discrete dynamical system or as
a discrete version of a quantum graph with localized wave packets.
%that have longest stabilization time are studied.
We study the set of such systems over metric graphs that can be constructed
from a given set of commensurable edges with fixed lengths.
It is shown that there always exists a system consisting of a bead graph with
vertex degrees not greater than three that demonstrates the longest stabilization
time in such a set.
The results are extended to graphs with incommensurable edges
using the notion of $\varepsilon$-nets and,
also, it is shown that dynamical systems of points on linear
graphs %with one initial point
have the slowest growth of the number of dynamic points.
%
%structural properties of dynamic-point systems
%%over such metric graphs
%that demonstrate the longest stabilization time are considered.
%It is shown that,
%in such a set of dynamical systems with the longest stabilization time,
%there is a system consisting of a bead graph
%with vertex degrees not greater than three.
\end{abstract}

%ELS

%=======================================================================

\section{Introduction}   
%================================================================================
% Introduction
%The case of quantum graph with narrow localized wave packets was recently
%studied in
%\cite{Chernyshev10_TimeDepSchrodingerEqStatisticsGaussianPacketsMG,
%ChernyshevTolchennikovShafarevich16_BehaviorQuasiParticlesHybridSp_RelGeomGeodesicsAnalyticNT}.

A dynamical system of points ($\mathit{DP}$-system) moving along the edges of
a metric graph,
yet its dynamics has a discrete nature, %that are studied in this paper
could be considered as a simplified discrete model of a quantum graph with
narrow localized wave packets \cite{Chernyshev10_TimeDepSchrodingerEqStatisticsGaussianPacketsMG,
ChernyshevTolchennikovShafarevich16_BehaviorQuasiParticlesHybridSp_RelGeomGeodesicsAnalyticNT}.
A quantum graph is a metric graph equipped with functions on its edges,
a differential operator acting on such functions, and matching conditions
on its vertices \cite{BerkolaikoKuchment13_IntroQuantum}.
Quantum graphs occurred as a model or tool in a number of problems
in chemistry, physics, engineering, and mathematics since 1930s
\cite{Pauling36_DiamagneticAnisotropyAromMol,Exner20_TopologicalBulkEdgeEffect}.
There exists a correspondence between the statistics of localized solutions
on a quantum graph and the dynamics of a $DP$-system.
Points in such a system may represent supports of Gaussian wave packets in a
quantum graph and/or projection of wave propagation on medium geodesics
\cite{ChernyshevTolchennikovShafarevich16_BehaviorQuasiParticlesHybridSp_RelGeomGeodesicsAnalyticNT}.

Some results towards the characteristics of the dynamics of such systems
were recently obtained in
\cite{TolchennikovChernyshev10_AsymptoticProperties,
%Chernyshev14_AsymptoticEstimateGaussianPackets,
Chernyshev18_AsymptoticCountingDynamicalDecorated,
Chernyshev17_CorrectionCountingMetricTree,
Chernyshev17_SecondTermNumberPointsMetricGraph}.
The growth of the number of points moving along edges and its asymptotics
are studied for metric trees in \cite{Chernyshev17_CorrectionCountingMetricTree}
reducing the counting problem to counting the number of lattice points
in an expanding polyhedra using Barnes multiple Bernoulli polynomials;
and, for some special cases, in \cite{ChernyshevTolchennikovShafarevich16_BehaviorQuasiParticlesHybridSp_RelGeomGeodesicsAnalyticNT}.
Further details and motivation for studying such systems, that emerges from
mathematical physics and other fields,
can be found in texts \cite{BerkolaikoKuchment13_IntroQuantum}
and \cite{Berkolaiko06_QuantumGraphsApps}.

Recently, in \cite{dworzanski2020CorrespondenceMGandTaPN}, it was shown that,
for a subclass of Timed-Arc Petri nets ($TaPN$-nets),
it is possible to overapproximate the number of timers with different values
in a $TaPN$-net by the number of points in a $DP$-system;
thus, asympotical estimations obtained for the number of points in a $DP$-system in
\cite{Chernyshev18_AsymptoticCountingDynamicalDecorated, Chernyshev17_CorrectionCountingMetricTree}
could be used to analyse behavioural properties related to timers in $TaPN$-nets.
Additionaly, the translation of a $DP$-system into a $TaPN$-net was developed
and implemented in TAPAAL to permit the analysis of fine TCTL properties of
a $DP$-system.

%TODO: Review before sending
Additional motivation comes from the problem of upper bounds on the
stabilization time for a $\mathit{DP}$-system on an arbitrary
metric graph.
While there exists asympotical estimates for the growth of the number of points
in a $DP$-system on an arbitrary metric graph,
estimates for stabilization time of a $DP$-system are obtained only for
two mostly trivial cases --- star graphs and graphs with edges of equal length.
In this paper, we develop an approach based on graph surgery to show
that, in the set of all dynamic-point systems with longest stabilization time
($\mathit{LSTDP}$-systems) constructed from a fixed set of edges,
there exists a $\mathit{DP}$-system on a metric graph with a specific structure
(bead graphs) having a dynamic point at one of its terminal vertices.
The results allow to take into account only the set of its edges and
$\mathit{LSTDP}$-systems of a specific structure to study
longest stabilization time.

Section~2 contains basic notions and definitions.
In section~3, it is shown that longest stabilization time could not always
be achieved using only tree metric graphs;
also, we introduce notions of point-places and walk classes used in
the next section.
Section~4 demonstrates the existence of a bead graph with vertex degrees
not higher than three among $\mathit{LSTDP}$-systems.
Section~5 extends the result of the previous section to metric graphs
with incommensurable edge lengths.
Section~6 concludes the paper with some further directions.

%=======================================================================

%=======================================================================
%\input{parts/2_motivating_example.tex}
%=======================================================================

%=======================================================================
\section{Preliminaries}
\label{sec:preliminaries}

%By $\zNbb$ we denote the set of non-negative natural numbers.
%For a set $S$, a \emph{bag (multiset)} $m$ over $S$ is a mapping
%$m:S\rightarrow \Nat$.
%The set of all bags over $S$ is denoted by $\Nat^S$.
%We denote addition and subtraction of two bags by $+$ and $-$,
%the number of all elements in $m$ taking into account the multiplicity
%by $\| m \|$, and comparisons of bags by $=, <, >, \leq, \geq$.
%that are defined as
%$m_1 \mathrel{R} m_2 \equiv \forall s \in S : m_1(s) \mathrel{R} m_2(s)$
%where $\mathrel{R}$ is one of $=,<,>,\leq,\geq$.
%We overload the set notation writing $\emptyset$ for the empty bag and $\in$
%for the element inclusion.

A metric graph $\Gamma$ is a graph consisting of set of vertices $V$,
set of undirected edges $E$,
and length function $l$ mapping each edge $e=\{ v_{1},v_{2}\} \in E$
to a positive real,
i.e., $l: E \to \mathbb{R_+}$ \cite{BerkolaikoKuchment13_IntroQuantum}.
For technical convenience, %when needed,
each undirected edge $e=\{ v_{1},v_{2}\}$ in $E$ may be considered as
a pair of arcs $\langle v_{1},v_{2}\rangle$ and $\langle v_{2},v_{1}\rangle$,
and both arc lengths coincide with the length of $e$;
both notions will be used in the paper interchangeably to
shorten some lengthy technical explanations.

The arc opposite to arc $a=\langle v_j,v_i \rangle$ is denoted by $\bar{a}$,
i.e., $\bar{a}=\langle v_i,v_j \rangle$.
For two points $x$ and $y$ on the graph, metric $\rho(x,y)$ is the shortest distance
between them, where distance is measured along the edges of the graph
additively.
A walk is a finite or infinite sequence of edges (arcs for directed graphs)
which joins a sequence of vertices.
A trail (path) is a walk in which all edges (vertices) are distinct.

The set of all walks from vertex $v$ to vertex $v'$ is denoted by
$\mathcal{W}(v,v')$.
For a walk $w$, the length $l(w)$ is the sum of lengths of all the arc entries in $w$.
The support of walk $w$ is the set of all arcs in $w$ and is denoted by $S(w)$.
The (directed) multisupport of walk $w$ in graph $\Gamma$ is a new (directed)
graph $\overline{S}(w)$ obtained from $\Gamma$ such that it has the same vertices $V$ as
$\Gamma$, and, for each entry of arc $\langle v_i, v_j \rangle$ in $w$,
we introduce a new edge $\{ v_i, v_j \}$
(new arc $\langle v_i, v_j \rangle$) into $\overline{S}(w)$,
i.e., the number of edges (arcs) in $\overline{S}(w)$ is equal to the
number of arc entries in $w$.
Note that a walk may contain multiple entries of an arc.

%an interval $[0,L_e]$, and $x_{e}$ is the coordinate on the interval
%. 
%%%Vertex $v_{1}$ corresponds to $x_e=0$, and $v_{2}$ to $x_e=L_e$,
%%%or vice versa.
%%%The choice of which vertex lies at zero is arbitrary with the alternative
%%%corresponding to a change of coordinate on the edge.

\begin{figure}[!t]
    \begin{center}
        \centerline{\includegraphics[scale=1]{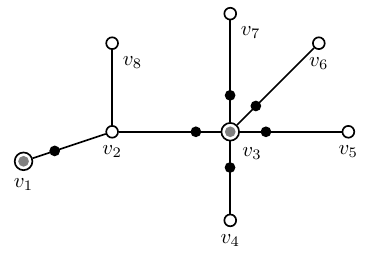}}
        \caption{System of dynamic points $\zPcal_\Gamma$
          on metric graph $\Gamma$}
        \label{fig:mg}
    \end{center}
\end{figure}

For a metric graph $\Gamma$, the dynamics of a system of dynamic points
$\mathcal{P}_\Gamma$ on $\Gamma$
is defined as following.
In the initial state, some vertices of $\Gamma$ hold a dynamic point.
When time starts to flow, each such point $p$ located in vertex $v$,
for each edge $e$ incident to $v$,
produces a point $p'$ on each $e$, and $p$ disappears
(intuitively, this corresponds to wave packet scattering);
each produced point $p'$ starts moving along corresponding~$e$.
Note that if we consider $e$ as a pair of directed arcs,
then $p'$ is generated on and moving along the arc outgoing from $v$
respecting the arc direction.
All points move with the same constant velocity; and,
due to new points generation, some arcs may carry more than one point.
When a moving point reaches vertex $v'$, again,
on each outgoing arc incident to $v'$, a new point is generated.
When more than one points reach a vertex simultaneously at $t$,
on each outgoing arc, only one point is produced,
as if only one point has reached the vertex at $t$;
i.e., points met on a vertex fuse, and each coordinate of an arc can carry
only one dynamic point.
However, points do not collide anywhere on edges except vertices,
i.e., if two points met on an edge,
they both continue their movement towards their own directions.
This becomes clearer if we consider the edge as a pair of arcs;
then, the points, converging on an edge, are moving along separate opposite arcs.
In Figure~\ref{fig:mg}, the initial set of points consists of two points in
vertices $v_1$ and $v_3$.
The point in $v_1$ produces a new point on edge $\{ v_1,v_2 \}$,
The point in $v_3$ produces points on edges leading to vertices
$v_2,v_4,v_5,v_6,v_7$.
After a time unit, there are no points in $v_1$ and $v_3$ (coloured gray),
but there are points (coloured black) moving from $v_1$ and $v_3$ 
to their adjacent vertices.

More examples and further details on $\mathit{DP}$-systems on metric graphs
and some of their extensions can be found in
\cite{Chernyshev17_SecondTermNumberPointsMetricGraph,
Chernyshev17_CorrectionCountingMetricTree,
ChernyshevTolchennikovShafarevich16_BehaviorQuasiParticlesHybridSp_RelGeomGeodesicsAnalyticNT}.

%=======================================================================

%=======================================================================
%== Linear metric graphs with slowest growth rate
\section{Cutting of a metric graph}

The number of dynamic points in $\mathcal{P}_\Gamma$ at time $t$ is denoted
by $N_{\mathcal{P}_\Gamma}(t)$.
%Note that the growth of $N_{\mathcal{P}_\Gamma}(t)$ at time $t_0$ is possible
%iff the number of points simultaneously reached vertex $v$ at $t_0$
%is strictly less than the (out-) degree of $v$.
The number of points on edge $e$ at time $t$ is denoted by $N_e(t)$.
For dynamical systems of points $\mathcal{P}_{\Gamma}$ and
$\mathcal{P}'_{\Gamma'}$,
we say that the growth rate of $\mathcal{P}_\Gamma$ is equal or less than
thus of $\mathcal{P}'_{\Gamma'}$ if
$\forall t \in \mathbb{R}_+{\setminus}\mathit{Coll} :
  N_{\mathcal{P}_\Gamma}(t) \le N_{\mathcal{P}'_{\Gamma'}}(t)$,
where $\mathit{Coll}$ is the (countable) set of time points when
more than one dynamic points meet on a vertex; we exclude such time points as
the number of dynamic points decreases for these moments, technically.
In what follows,
we discuss growth and stabilization implicitly omitting vertices collision
time points $\mathit{Coll}$.
  
The stabilization time $t_s(\mathcal{P}_\Gamma)$ of $\mathit{DP}$-system
$\mathcal{P}_\Gamma$ on graph $\Gamma$ with edges of commensurable lengths
is the value of the period of time from the initial time point to
the point in time when the number of dynamic points $N_{\Gamma}(t)$ on $\Gamma$
has been stabilized
\cite{ChernyshevTolchennikovShafarevich16_BehaviorQuasiParticlesHybridSp_RelGeomGeodesicsAnalyticNT},
\textit{i.e.},
$\forall t \in \mathbb{R}_+{\setminus}\left(\left[0,t_s\left(\mathcal{P}_\Gamma\right)\right)
\cup\mathit{Coll}\right)
: N_{\Gamma}(t) = N_{\Gamma}(t_s(\mathcal{P}_\Gamma))$.
For a given set of edges $E$, let $\mathit{LSTDP}(E)$ be the set of $\mathit{DP}$-systems
constructed from $E$
demonstrating the longest stabilization time ($\mathit{LSTDP}$-systems);
$\mathit{LSTDP}(E)$ is not necessarily a singleton.

Number $N_{\mathcal{P}_\Gamma}(t)$ increases at time $t_0$
when the number of points simultaneously reached a vertex $v$ at $t_0$
is strictly less than the (out-) degree of $v$.

In this section, it is shown that, while for an arbitrary $\mathit{DP}$-system
$\mathcal{P}_\Gamma$,
there is always a tree $\mathit{DP}$-system which growth rate is equal or less than
thus of $\mathcal{P}_\Gamma$;
set $\mathit{LSTDP(E)}$ does not always contain a $\mathit{DP}$-system
on a tree.

%for any dynamical system $\mathcal{P}_\Gamma$ on
%a graph $\Gamma$ constructed from set of edges $E$,
%there is a $\mathcal{P}'_{\Gamma'}$ consisting of linear graph $\Gamma'$
%constructed from $E$ and one dynamic point at its terminal vertex,
%such that the growth rate of $\mathcal{P}'_{\Gamma'}$ is slower than
%thus of $\mathcal{P}_\Gamma$.
%However, a longest stabilization time DP-system does not always
%contains a linear metric graph.

\begin{figure}[!t]
    \begin{center}
        \centerline{\includegraphics[scale=1]{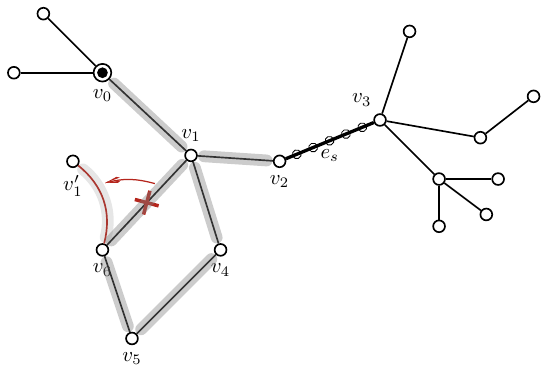}}
        \caption{Metric graph with a cycle $\langle v_1, v_4, v_5, v_6 \rangle$}
        \label{fig:NP_cycleelim}
    \end{center}
\end{figure}

To begin with, we suggest to partition walks into equivalence classes
under the following equivalence.
Let $v$ be a vertex and $e=\{ v_a, v_b \}$ be an edge of $\Gamma$.
Let set $\zWcal(v, e)$ be the union of sets $\zWcal(v, v_a)$ and $\zWcal(v, v_b)$.
Let us introduce an equivalence relation $\approx$ over set of walks
$\zWcal(v,e)$ defined as follows:
\begin{itemize}
	\item if walks $w_1$ and $w_2$ both end in $v_a$ or both end in $v_b$,
then they are equivalent iff their lengths are congruent modulo $2l(e)$, \textit{i.e.},
$$%
%\begin{multline*}
 \forall \langle w_1,w_2 \rangle \in \zWcal(v,v_a)^2
\cup \zWcal(v,v_b)^2:%\\
w_1 \approx w_2 \iff l(w_1)\equiv l(w_2) \Mod{2l(e)}
%\end{multline*}
$$
  \item if, w.l.o.g., $w_1$ ends in $v_a$, and $w_2$ ends in $v_b$,
then they are equivalent iff the difference of their lengths is congruent
to $l(e)$ modulo $2l(e)$, \textit{i.e.},
$$\forall \langle w_1,w_2 \rangle \in \zWcal(v,v_a) \times \zWcal(v,v_b):
w_1 \approx w_2 \iff l(w_1)-l(w_2) \equiv l(e) \Mod{2 l(e)}$$
%%%\begin{multline*}
%%%\forall \langle w_1,w_2 \rangle \in \zWcal(v,v_a) \times \zWcal(v,v_b):
%%%%\forall w_1\in W(v,v_a): \forall w_2 \in W(v,v_b):
%%%\\ w_1 \approx w_2 \iff l(w_1)-l(w_2) \equiv l(e) \Mod{2 l(e)}
%%%\end{multline*}
\end{itemize}

%\begin{multline}
	%\sum\limits_{p\in \zpreset{t_\zsmSN}}{w_t(p) \cdot \|\gamma(\langle p, t_\zsmSN \rangle)\|}\\
	%= \sum\limits_{p \in \zpostset{t_\zsmSN}} \left( {  w_t(p) \cdot \|\gamma(\langle t_\zsmSN, p\rangle)\|}
	 %+ w_m(p) \cdot \quad \sum\limits_{\mathclap{c \in Con(\langle t_\zsmSN, p \rangle)}}{
			%\enskip \zhWcal(p)(c)} \right),
	%\label{eqn:SA1}
%\end{multline}
For each edge $e=\{v_1,v_2\}$, set $\zWcal(v_0,e)$ is partitioned into
equivalence classes under~$\approx$.
The arrival of point $p$ at vertex $v$ of $e$,
moving from $v_0$ to $v$ along walk $w$,
induces a new point on $e$ iff $w$ has the minimum length for its class
$[w]_{\approx}$; \textit{i.e.}, only the shortest walks of $[w]_{\approx}$
induce a new point on $e$.
Now, consider two arcs $\langle v_1,v_2 \rangle$ and $\langle v_2,v_1 \rangle$
as a set of point-places forming a loop;
the point-places move around the loop along the arc directions with the same
velocity as dynamic points.
We name them point-places as, when a dynamic point reaches a point-place,
it continues its movement bound to the point-place;
we denote point-places that will be reached by dynamic points with empty circles
(as in Fig.~\ref{fig:NP_cycleelim} on edge $e_s$).
Thus, for edge $e$, an equivalence class $[w]_{\approx}$ of $\zWcal(v_0,e)$
corresponds to a point-place which will be saturated by a dynamic point
that came from $v_0$ along the minimal path of $[w]_{\approx}$.
It is possible to define point-place evolution and corresponding
time-dependent equivalence classes in a time-dependent manner;
to avoid handling time-dependent coordinates, we just take
class representatives for $t=0$.

The first quite obvious observation is that, for any dynamical system
$\mathcal{P}_\Gamma$,
there is a one-point $\mathit{DP}$-system on a tree constructed from the same
edges that has less or equal growth rate.
We consider only connected $\mathit{DP}$-systems with at least one dynamic point.
For disconnected graphs,
stabilization can be considered for its components independently.

\begin{lemma}
Let $E$ be a given set of edges with a length function $l$ on $E$,
and $\mathcal{P}_\Gamma$ be a dynamical system with a graph $\Gamma$
constructed from $E$.
Then, there is a dynamical system $\mathcal{P}'_{\Gamma'}$ consisting of
a tree graph $\Gamma'$ constructed from $E$ and one dynamic point,
such that the growth rate of $\mathcal{P}'_{\Gamma'}$ is equal or slower than
thus of $\mathcal{P}_\Gamma$.
\label{lemma:slowtree}
\end{lemma}

\begin{proof}
Let $\mathcal{P}_\Gamma$ have $N_0$ initial points on metric graph $\Gamma$.
Obviously, elimination of points in $\mathcal{P}_\Gamma$ can only decrease
growth rate.
Thus, if $N_0 > 1$, it is safe to eliminate all points in $\zPcal_\Gamma$
except one arbitrary dynamic point $p_0$ in vertex $v_0$.

Let $\Gamma$ be not acyclic.
The proof is based on an observation that a cutting of a cycle
in $\Gamma$ can only decrease growth rate.
Let
$\langle v_{i_1},v_{i_2} \rangle=a_1,\ldots, a_k=\langle v_{i_{k}},v_{i_1}\rangle$
be a cycle in $\Gamma$.
Let us split the cycle at arc $a_k$, \textit{i.e.}, a vertex $v_{i_1}'$
is introduced,
arcs $a_k$ and $\bar{a}_k$ are replaced with $a_k'=\langle v_k,v_{i_1}'\rangle$ and
$\bar{a}_k'$ in resulting graph $\Gamma'$;
in Figure~\ref{fig:NP_cycleelim}, cycle $\langle v_1, v_4, v_5, v_6 \rangle$
is split at vertex $v_1$, \textit{i.e.}, $\{ v_6,v_1\}$ is removed, 
vertex $v_1'$ and edge $\{v_6,v_1'\}$ are added.
For any edge $e$, the operation can only narrow set $\zWcal(v_0,e)$;
\textit{i.e.}, all walks of the resultant graph $\Gamma'$ are realizable in $\Gamma$.
%%%namely, for any walk in $\Gamma'$, obtain its greedy version, and
%%%consider it in $\Gamma$.
Thus, no shorter walks are introduced;
and, no point-place on $e'$ in $\Gamma'$ is saturated earlier than
the corresponding one on $e$ in original $\Gamma$.

By repeating cycle elimination, we obtain $\mathcal{P}'_{\Gamma'}$ on
a tree metric graph.
%\qed
\end{proof}

\begin{figure}[!t]
    \begin{center}
        \centerline{\includegraphics[scale=0.8]{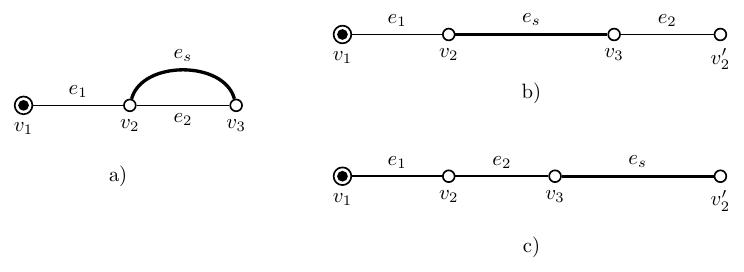}}
        \caption{Metric graph $\Gamma$ at a) and its possible cuttings
          $\Gamma'$ at b) and $\Gamma''$ at c)}
        \label{fig:MG_stabtime}
    \end{center}
\end{figure}

However, we cannot use the same approach to achieve an $\mathit{LSTDP}$-system
on a tree.
A cutting may completely eliminate some classes of walks
of original $\zPcal_\Gamma$ that stipulate longest stabilization time
($\mathit{LST}$-classes),
\textit{i.e.}, those whose shortest walks are longest ($\mathit{LST}$-walks)
among all
the walk classes of $\zPcal_\Gamma$ for all the edges of $\Gamma$.
After cutting, the dynamic points corresponding to such $\mathit{LST}$-classes never occur in
$\zPcal_\Gamma$,
so stabilization does not depend on (wait for) these dynamic points
and may occur earlier.
%There are several ways to cope with this;
%
The metric graph $\Gamma$ that demonstrates such phenomenon is depicted in
Figure~\ref{fig:MG_stabtime}a).
The dynamic point system $\zPcal_\Gamma$ on $\Gamma$ consists of only one dynamic
point in $v_1$. Lengths of $e_1$ and $e_2$ are equal to $1$, and thus of $e_s$
is equal to $2$.
The stabilization time $t_s(\zPcal_\Gamma)$ is equal to $4$,
and $\forall t > t_s(\zPcal_\Gamma): N_{\Gamma}(t) = 8$.
If we cut $\Gamma$ into $\Gamma'$ as in Figure \ref{fig:MG_stabtime}b),
stabilization time $t_s(\zPcal_\Gamma')$ is downgraded to $3$, and
$\forall t > t_s(\zPcal_\Gamma'): N_{\Gamma'}(t) = 4$.

To handle this obstacle, we suggest to fix one of shortest walks in
an $\mathit{LST}$-class, \textit{i.e.,} an $\mathit{LST}$-walk,
and conduct the cutting preserving the walk.
In the example above, the `missed' points were generated by the eliminated
walks in the class that contains $\mathit{LST}$-walks $e_1, e_1, e_1, e_2$ and
$e_1, e_2, e_2, e_2$.
Thus, if we cut graph $\Gamma$ into $\Gamma''$ preserving such walks,
then $\zPcal_{\Gamma''}$ will have stabilization time not less
than $\zPcal_\Gamma$ even if some non-$\mathit{LST}$-classes vanish and
overall $N_\Gamma$ could decrease.
A possible cutting is depicted in Figure~\ref{fig:MG_stabtime}c).

\begin{figure}[!t]
    \begin{center}
        \centerline{\includegraphics[scale=0.8]{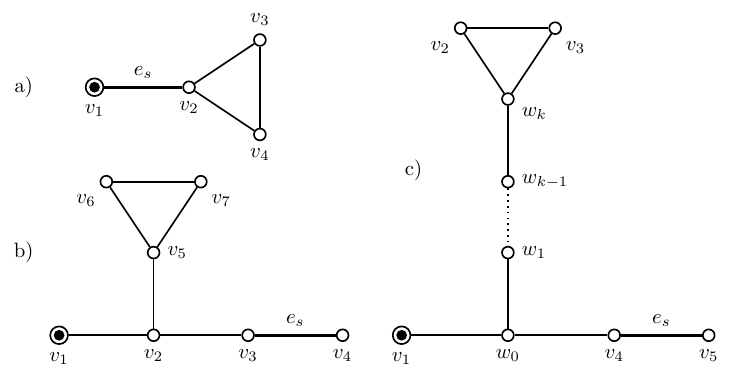}}
        \caption{Examples of $\mathit{DP}$-systems that have stabilization time greater
          than any acyclic graph with the same set of edges.
          All edges have length one.}
        \label{fig:MG_uncuttable}
    \end{center}
\end{figure}

If such a cycle cutting preserving $\mathit{LST}$-walk would always be found,
then $\mathit{LSTDP}$ should always contain a tree and, moreover, using
a procedure from the next section, a linear graph.
Unfortunately, the graphs of $\mathit{DP}$-systems in Figure~\ref{fig:MG_uncuttable}
cannot be cut into an acyclic graph preserving the stabilization time.
All the edges of the graphs in Figure~\ref{fig:MG_uncuttable} have length
equal to 1.
For $\mathit{DP}$-system in \ref{fig:MG_uncuttable}a), it could be checked manually;
the last point is generated on $e_s$ at $t=4$, while, for a linear graph with
4 edges, stabilization time is $3$, for other trees --- even less.
Figure \ref{fig:MG_uncuttable}b) shows that an uncuttable cycle could lie
not on a path from initial point vertex $v_1$ to stabilization edge $e_s$.
Figure \ref{fig:MG_uncuttable}c) shows that the absolute difference between
the stabilization time of $\mathit{LSTDP}$-system on \ref{fig:MG_uncuttable}c)
($t_s=2k+5$) and a linear graph with the same set of edges ($t_s=k+5$)
could be arbitrary large.

%== End of Linear metric graphs with slowest growth rate

%=======================================================================
%== Beads metric graphs with longest stabilization time
\section{Bead $\mathit{DP}$-systems with longest stabilization time}
In this section, by applying $\mathit{LST}$-path preserving operations,
we show that set $\mathit{LSTDP}(E)$ always contains a $\mathit{DP}$-system over
a graph with a specific structure.

While it is not hard to see that a star metric graph has the stabilization time
not more than two times greater than the shortest stabilization time
for a fixed set of edges,
structure of $DP$-systems that demonstrate longest stabilization time
is not yet understood.
Also, $\mathit{DP}$-systems on a metric graph with incommensurable edges never
stabilizes and, moreover, ergodic 
\cite{ChernyshevTolchennikovShafarevich16_BehaviorQuasiParticlesHybridSp_RelGeomGeodesicsAnalyticNT};
therefore, we consider stabilization time for metric graphs with commensurable
edges only.

A graph is called a bead graph if it does not contain incident cycles,
\textit{i.e.}, there is no vertex that belongs to two different cycles.
It is shown below that, among $\mathit{LSTDP}(E)$, there is a bead
$\mathit{LSTDP}$-system with one dynamic point at one of its terminal
vertices and vertex degrees not greater than three.
The argument is done in three steps.
The first claim is that, for any $E$, there is a bead $\mathit{DP}$-system
in set $\mathit{LSTDP}(E)$.

\begin{lemma}
Let $E$ be a given set of edges with a length function $l$ on $E$.
Set $\mathit{LSTDP}(E)$ contains a bead $\mathit{LSTDP}$-system $\Gamma$
with a dynamic point at a vertex.
\label{thm:DP_bead}
\end{lemma}
\begin{proof}
By scaling, \textit{i.e.}, dividing by $gcd$, w.l.o.g.,
the lengths of all edges in $E$ can be made integer with $gcd(E)$ equal to $1$.

%Proof of one point
The points generated by dynamic point $p$ reached a vertex are called
descendants of $p$, and $p$ is the ancestor of the new points;
descendants and ancestor are transitive notions, \textit{i.e.}, descendants of
descendants of $p$ are descendants of $p$.
The stabilization time for an edge $e$ is denoted by $t_s(e)$, \textit{i.e.},
$t_s(e)$ is the earliest time point such that
$\forall t> t_s(e): N_e(t) = N_e(t_s(e))$.
$N_e(t)$ stabilizes when edge $e$ receives $N_e(t_s(e))$ points.
Assume that $e_s$ is the edge with the longest stabilization time $t_s$
in $\zPcal_\Gamma$.
Now consider the last point $p_s$ that is generated on $e_s$ at stabilization
time $t_s(e)$.
Point $p_s$ is the descendant of an initial point $p_0$ in $\zPcal_\Gamma$.
Elimination of all initial points except $p_0$ in $\zPcal_\Gamma$
does not decrease $t_s(\zPcal_\Gamma)$.
Thus, we may consider only $\mathit{DP}$-systems with one dynamic point.

%%%%%%Let $\mathit{LSTDP}(E)$ contains $\zPcal_\Gamma$ consisting of graph
%%%%%%$\Gamma$ and a dynamic point in vertex $v_0$,
%%%%%%and $\Gamma$ is not acyclic.
%%%%%%The proof is based on an observation that there is a cutting of loops
%%%%%%in $\Gamma$ that does not decrease $t_s(\zPcal_\Gamma)$.

Let $\zPcal_\Gamma$ be a $\mathit{DP}$-system on graph $\Gamma$
in set $\mathit{LSTDP}(E)$ with point $p$ in $v_0$.
Let $e_s$ be the edge with the longest stabilization time $t_s$ in
$\zPcal_\Gamma$, and let $w_s \in {\zWcal}(v_0,e_s)$ be an $\mathit{LST}$-walk.

It is clear that, for any walk $w$ from vertex $u$ to $v$, it is possible to
obtain a new walk $w'$ with the same length by reordering arc entries in $w$,
such that, when any entry of arc $a_p=\langle v_i,v_j \rangle$,
that corresponds to edge $e_p$, is met in $w$ for the first time,
walk $w'$ runs back and forth on $e_p$ (do alternating series of steps
$\bar{a}_p$ and $a_p$) until there are no more entries of arcs $\bar{a}_p$
or $a_p$ left in $w$,
except probably one last entry to preserve ability to reach $u$.
Intuitively, consider multisupport $\overline{S}(w)$,
which is semieulerian by construction,
and apply `greedy' modification of
Fleury algorithm~\cite{Fleischner1990_EulerianGraphs_ChapX},
that always chooses just passed edge again if it may,
to find a semieulerian walk from $u$ to $v$ on $\overline{S}(w)$;
it is possible as Fleury algorithm is path-choice agnostic
unless the edge is a bridge.
For example, for walk $w= a_2, a_3, \bar{a}_3, \bar{a}_2, a_2, a_3, a_4$,
there is a walk $w'= a_2, \bar{a}_2, a_2, a_3, \bar{a}_3, a_3, a_4$;
we will call such reordered walks \emph{greedy}.

For a walk $w$, if the number of arcs in $\overline{S}(w)$ corresponding to edge
$e$ in $\Gamma$ is odd (even),
we will call edge $e$ in $\Gamma$ and the corresponding arcs in $\overline{S}(w)$
\emph{odd} (\emph{even}).
%%%%While constructing walk $w'_s$ from $w_s$, for reasons that will become clearer
%%%%below, we will always choose to traverse even edges before odd.

Let $w'_s \in {\zWcal}(v_0,e_s)$ be the greedy version of $\mathit{LST}$-walk
$w_s$,
and, obviously, $w'_s$ is an $\mathit{LST}$-walk in ${\zWcal}(v_0,e_s)$ itself.
Note that if, while traversing $\overline{S}(w'_s)$ according to $w'_s$,
we reached arcs corresponding
to an edge $e$ of $\Gamma$ which multiplicity is larger than $2$,
then, by the greedy Fleury procedure, we always move back and forth on these arcs,
except the case when the last arc becomes a bridge in $\overline{S}(w'_s)$
during such traversing.
Thus, while analysing the structure of a semieulerian walk for graph cutting
purposes, we may temporary omit all
`parasite' pairs of arcs in $\overline{S}(w'_s)$ leaving only one arc for odd and
two arcs for even edges;
%(\textit{i.e.}, for $k$ arcs, we keep $2^{k+1(\mathrm{mod}\;{2})}$ of them);
such a multisupport is called \emph{reduced} and
is denoted by $\widehat{S}(w'_s)$.
To preserve $w'_s$ length, we need to restore such omitted arcs after
graph cuttings.
%(\textit{i.e.}, for $k$ parity we leave $2^{(k+1)\;\mathrm{mod}\ 2}$ arcs)

\begin{figure}[!ht]
    \begin{center}
        \centerline{\includegraphics[scale=1]{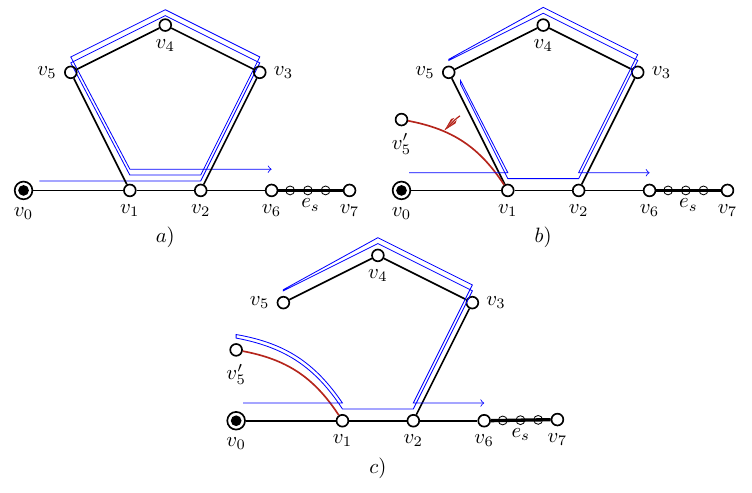}}
        %\centerline{\includegraphics[scale=1]{Fig5.pdf}}
        \caption{a) Semieulerian $v_0v_6$-walk, b) its greedy version,
            and c) graph $\Gamma$ after cutting ---
            edge $\langle v_1, v_5 \rangle$ is split from $v_5$}
        \label{fig:cut1_greedy}
    \end{center}
\end{figure}

As we will see now, greedy semieulerian walks have simpler
structure than arbitrary ones.  
At first, we take greedy walk $w'_s$ and start traversing $\overline{S}(w'_s)$
according to $w'_s$.
If walk $w'_s$ reaches vertex $v_a$ with an incident even
edge $e=\langle v_a,v_b \rangle$,
and, after traversing back and forth all the arcs corresponding to $e$,
it ends up in $v_a$ (\textit{i.e.}, it doesn't passes beyond $v_b$),
then we split $e$ from vertex $v_b$.
Such case may happen only when all but last arcs on $e$
are traversed and the last arc is not a bridge,
\textit{i.e.}, even if we traverse-remove the last arc $\langle v_b, v_a \rangle$
corresponding to $e$ in $\overline{S}(w'_s)$,
there are still paths to the rest arcs incident to $v_b$ in the untraversed
fragment of $\overline{S}(w'_s)$;
this implies that when we first reached $e$ before splitting, it lies on a cycle in the
untraversed part of $\overline{S}(w'_s)$.
Figure~\ref{fig:cut1_greedy} provides an example of such cutting.
We have walk $w_s=v_0,v_1,v_2,v_3,v_4,v_5,v_1,v_2,v_3,v_4,v_5,v_1,v_2,v_6$,
its greedy version $w'_s=v_0,v_1,v_5,v_1,v_2,v_1,v_2,v_3,v_4,v_5,v_4,v_3,v_2,v_6$,
and a cutting of edge $e=\langle v_1,v_5 \rangle$ from vertex $v_5$ as $w'_s$
reaches but does not cross (walk beyond) $v_5$ from $e$.
After all cuttings, we obtain resultant graph $\Gamma'$.
Obviously, the cutting procedure preserve $w'_s$; \textit{i.e.}, $w'_s$ is realizable
in $\Gamma'$.
All cycles that contain even edges are cut, \textit{i.e.},
there are not cycles with even edges in $\Gamma'$.
Thus, even edges will correspond to bridges in $\Gamma'$ after cutting.

Note that we do cuttings 
on $\overline{S}(w'_s)$ and the original graph simultaneously,
even we are focused only on $\overline{S}(w'_s)$ during the procedure.
Multisupport $\overline{S}(w'_s)$ plays auxiliary role to let us see how to cut the
original graph preserving the $\mathit{LST}$-walk $w'_s$.

\begin{figure}[!ht]
    \begin{center}
        \centerline{\includegraphics[scale=1]{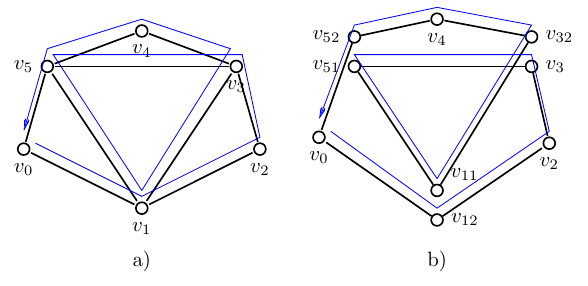}}
        %\centerline{\includegraphics[scale=1]{Fig6.pdf}}
        \caption{a) Walk $w_{cc}$ in component $CC$ of $\overline{S}_{mod2}(w'_s)$,
            and b) the resulting component after all cuttings conducted}
        \label{fig:cut2_graphmod2}
    \end{center}
\end{figure}

Now, we take $\Gamma'$ after cutting and $\overline{S}_{\Gamma'}(w'_s)$,
and, if there is an edge $e$ in $\Gamma'$ that corresponds to two or more arcs
in $\overline{S}_{\Gamma'}(w'_s)$, then we remove a pair of arcs that corresponds to $e$
and repeat removing while we can; the result is denoted by
$\overline{S}_{mod2}(w'_s)$.
Clearly, when the process ends, all arcs corresponding to even edges
will be removed, and, for each odd edge, only one arc is kept.
As parity has not changed, all components of odd edges are eulerian,
except one that corresponds to a semieulerian $v_0e_s$-tour.
For each eulerian connected component $CC$ in $\overline{S}_{mod2}(w'_s)$,
we fix eulerian walk $w_{cc}$ within $CC$.
For each vertex $v_i$ of $w_{cc}$ with degree $d(v_i)$ more than two,
\textit{i.e.}, $w_{cc}$ crosses vertex $v_i$ exactly $d(v_i)/2$ times,
we split $v_i$ into $d(v_i)/2$ vertices and make $w_{cc}$ pass through
different copies of $v_i$, thus, not crossing any such vertex $v_i$ more than once.
Therefore, by splitting and adding new vertices, we transform $CC$ into one cycle.
For the semieulerian component, we do the same process but the result is
a linear graph.
After cutting, we revive all deleted pairs of arcs on edges, that were removed
in $\overline{S}_{mod2}(w'_s)$, and restore path $w'_s$ in the resulting graph;
it does not introduce any difficulties as, despite we added new vertices,
the set of edges is the same --- for example, edge $\langle v_1,v_3 \rangle$
in Figure~\ref{fig:cut2_graphmod2}a) corresponds to edge
$\langle v_{11},v_{32} \rangle$ of the resulting graph in
Figure~\ref{fig:cut2_graphmod2}b).
For vertex $v_i$ in $\Gamma'$, there can be several vertices in $\Gamma''$;
for example, in Figure~\ref{fig:cut2_graphmod2}, for vertex $v_1$ in $\Gamma'$,
there are vertices $v_{11}$ and $v_{12}$ in $\Gamma''$.
While we reconstructing even edges that are incident to $v_i$ in $\Gamma'$,
it is irrelevant which of new vertices corresponding to $v_i$ in $\Gamma''$
we choose as our only focus is to preserve $w'_s$, even less ---
a path of the same length.

The resulting graph is denoted by $\Gamma''$ and path $w'_s$ reconstructed in
$\Gamma''$ is denoted by $w''_s$.
Graph $\Gamma''$ consists of a linear subgraph and a number of cycles
corresponding to odd edges,
that are connected with bridges corresponding to even edges;
to easier imagine it --- if each cycle of odd edges is contracted to a vertex,
the resulting graph is a tree of even edges.
By construction, $\Gamma''$ is a bead graph; all contacting cycles of odd 
edges are merged into one cycle, and
any vertex in $\Gamma''$ has not more than two incident odd edges
with exception of $e_s$.

As $w''_s$ has the same length as $w_s$, $w''_s$ is $\mathit{LST}$-walk,
and $\zPcal_{\Gamma''}$ is an $\mathit{LSTDP}$-system in $\mathit{LSTDP}(E)$.

\end{proof}

\begin{figure}[!ht]
    \begin{center}
        \centerline{\includegraphics[scale=1]{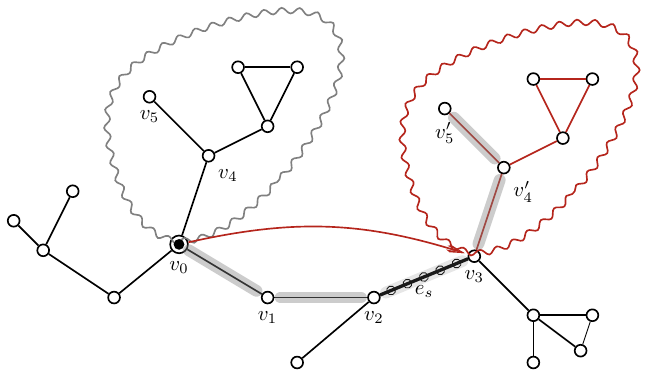}}
        %\centerline{\includegraphics[scale=1]{Fig7.pdf}}
        \caption{Moving the graph fragment from vertex $v_0$ to $v_3$}
        \label{fig:NP_movebeyond}
    \end{center}
\end{figure}

\begin{figure}[!ht]
    \begin{center}
        \centerline{\includegraphics[scale=1]{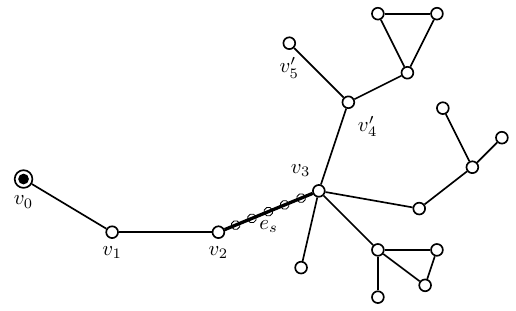}}
        %\centerline{\includegraphics[scale=1]{Fig8.pdf}}
        \caption{Resulting bead broom metric graph}
        \label{fig:NP_beadbroom}
    \end{center}
\end{figure}

A graph is called a bead broom graph if it is a bead graph with a fixed
connected (linear) subgraph --- a handle; all vertices of the handle
have at most degree two, except one of its terminal vertices.
For example, in Figure~\ref{fig:NP_beadbroom}, the graph is a broom graph
with a handle $v_0, v_1, v_2, v_3$.

\begin{lemma}
Let $E$ be a given set of edges with a length function $l$ on $E$.
If there is a bead $\mathit{LSTDP}$-system $\Gamma$ in $\mathit{LSTDP}(E)$
with a dynamic point at vertex $v_0$,
then $\mathit{LSTDP}(E)$ contains a bead broom $\mathit{LSTDP}$-system $\Gamma'$
with a point at the end of the handle of $\Gamma'$,
and the handle contains edge $e_s$ with the longest stabilization time.
\label{thm:DP_beadbroom}
\end{lemma}

\begin{proof}
Let $e_s$ be an edge with longest stabilization time.
We apply to $\Gamma$ the procedure described in the proof of
Lemma~\ref{thm:DP_bead},
and, after cutting, $\Gamma$ consists of a linear $v_0e_s$-subgraph,
odd cycles, and even bridges.

Let $e_s = \langle v_r, v_q \rangle$,
where $v_r$ is the vertex that is closer to $v_0$,
\textit{i.e.}, $v_q$ does not lie on the shortest path $h$ from $v_0$ to $v_r$.
Path $h$ is unique as $v_0$ and $v_r$ belong to the linear subgraph of odd edges
in $\Gamma$ constructed at the previous phase;
$h$ will become the handle of the resulting bead broom metric graph.

Consider now an arbitrary even edge $e_b = \langle v_a,v_b \rangle$ such that
$v_a$ incident to $h$.
Let subgraph $g$ of $\Gamma$ be the connected component containing
$v_b$ that appears if we remove bridge $e_b$
(even $e_b$ is a bridge as was shown in the the proof of
Lemma~\ref{thm:DP_bead}).
For example, on the left of Figure~\ref{fig:NP_movebeyond},
such a subgraph containing vertices $v_4$ and $v_5$ is outlined by a grey
wavy line.
Vertex $v_4$ belongs to $g$ and is adjacent to $v_0$.

New graph $\Gamma'$ is obtained by disconnecting $g$ from vertex $v_a$
in $\Gamma$ and connecting $g$ to $v_q$, \textit{i.e.},
$e_b$ is removed from $\Gamma$,
and a new edge $\langle v_b, v_q\rangle$ of length $l(e_b)$ is added.
The moved subgraph $g$ with the new edge in $\Gamma'$ will be denoted by $g'$,
and each vertex $v$ and each edge $e$ in $g$ will be denoted by $v'$ and $e'$,
correspondingly, in $\Gamma'$.
In Figure~\ref{fig:NP_movebeyond}, $g$ is moved from $v_0$ to $v_3$.
%Let set $W(v_0, e_s)$ be the union of sets $W(v_0, v_2)$ and $W(v_0, v_3)$.
Consider set $\zWcal(v_0, e_s)$ in $\Gamma$ and set $\zWcal(v_0, e_s)$ in $\Gamma'$.
It is needed to ensure that, after such surgery of $\Gamma$,
no new walks from $v_0$ to the endpoints of $e_s$
that decrease $t_s(e_s)$ appeared in $\Gamma'$.

Every walk $w'$ in $\zWcal(v_0, v_q)$ of $\Gamma'$ that has no edges of $g'$
is, clearly, in $\zWcal(v_0, e_s)$ of $\Gamma$.
Let $w'$ be a walk in $\zWcal(v_0, e_s)$ of $\Gamma'$ that has edges of $g'$
in its support $S(w')$.
For example,
consider path $w'=\langle v_0, v_1, v_2, v_3, v_4', v_5', v_4', v_3 \rangle$
in Figure~\ref{fig:NP_movebeyond};
its support is highlighted with gray.
For walk $w'$, there is a walk $w$ in $\Gamma$ that is not longer than $w'$.
Walk $w$ contains all edges of $w'$ that do not belong to $g'$;
in addition, for each edge $\langle v_i', v_j' \rangle$ of $w'$ in $g'$,
$w$ contains corresponding edge $\langle v_i, v_j \rangle$ in $w$.
For $\langle v_q, v_b' \rangle$ in $w'$, $w$ contains $e_b$ .
The lengths of such $w'$ and $w$ are equal;
and, thus, $w'$ and $w$ lie in the same class $[w']_{\approx}$.
Point $p$ moving along $w$ from $v_0$ to $v_q$ in $\Gamma$ can even induce a new point
on $e_s$ earlier than moving along $w'$ in $\Gamma'$ if $w$ does not contain edges of
$\Gamma$ beyond endpoint $v_q$ of $e_s$,
\textit{i.e.}, $p$ produces a new point on $e_s$ in $l(e_s)$ time units earlier.
Thus, for every walk $w'$ in $\Gamma'$, there is a walk in $\Gamma$ that
produces a point on $e_s$ not later than $w'$.
As the result, $e_s$ will be saturated in $\zPcal_{\Gamma'}$ not earlier than in
$\zPcal_\Gamma$, and $\zPcal_{\Gamma'}$ belongs to $\mathit{LSTDP}(E)$.%
\end{proof}

\begin{figure}[!ht]
    \begin{center}
        \centerline{\includegraphics[scale=1]{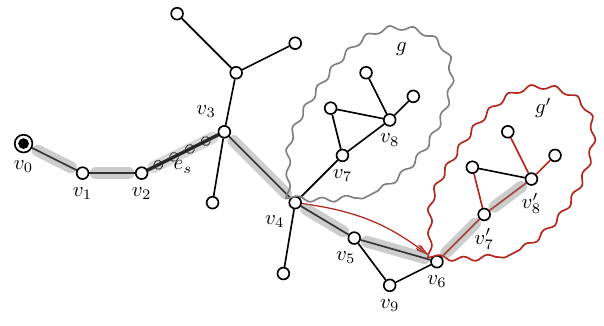}}
        %\centerline{\includegraphics[scale=1]{Fig9.pdf}}
        \caption{Relocation of subgraph $g$ from $v_4$ to $v_6$}
        \label{fig:th2_linearize_broom}
    \end{center}
\end{figure}

Eventually, we will show that vertex degrees of $\mathit{LSTDP}$-system can be
reduced to three or less without leaving $\mathit{LSTDP}(E)$ set.

\begin{lemma}
Let $E$ be a given set of edges with a length function $l$ on $E$.
If there is a bead broom $\mathit{DP}$-system $\zPcal_\Gamma$ in $\mathit{LSTDP}(E)$
with a handle $h$,
an initial point at the end of the handle terminal vertex $v_0$,
and $h$ contains an edge $e_s$ with longest stabilization time,
then $\mathit{LSTDP}(E)$ contains a bead $\mathit{DP}$-system $\zPcal_\Gamma'$
with an initial point at one of its terminal vertices and
maximum degree of 3.
\label{thm:beadthree}
\end{lemma}

\begin{proof}
Let $e_s = \langle v_r, v_q \rangle$,
where $v_q$ is the vertex that is farther from $v_0$ than $v_r$.

If a vertex $v$ is a leaf or $v$ belongs to a cycle $C$ in $\Gamma$ that has
only one incident bridge (even edge) $e_b$, \textit{i.e.}, $C$ is a terminal cycle,
and $v$ is not incident to $e_b$, then $v$ is called a \emph{bead leaf}
vertex.

Consider a path $p_l$ from $v_q$ to a bead leaf $v_l$ of $\Gamma$ that does not
run through $e_s$,
and an arbitrary even edge $e_b$ of $\Gamma$ that is incident to a vertex $v_a$ of
$p_l$. Let subgraph $g$ be the connected component that occurs if we remove
$e_b$, thus that does not contain $p_l$.
Let $v_b$ be another terminal vertex of $e_b$ adjacent to $v_a$.
New graph $\Gamma'$ is obtained by disconnecting $g$ from vertex $v_a$
in $\Gamma$ and connecting $g$ to $v_l$,
\textit{i.e.}, edge $e_b$ is removed and a new edge $\langle v_b, v_l\rangle$
of length $l(e_b)$ is added.
The moved subgraph $g$ with the new edge in $\Gamma'$ will be denoted by $g'$,
and each vertex $v$ and each edge $e$ in $g$ will be denoted by $v'$ and $e'$,
correspondingly, in $\Gamma'$.

For example, in Figure~\ref{fig:th2_linearize_broom}, there is path
$w_l = \langle v_3, v_4, v_5, v_6 \rangle$ from $v_3$ to bead leaf $v_6$.
Subgraph $g$ containing vertices $v_7$ and $v_8$ is outlined by a
grey wavy line and is incident to vertex $v_4$ of the path.
Graph $\Gamma'$ is obtained from $\Gamma$ by moving $g$ from $v_4$
to $v_6$.

The second part of the proof argument mostly resembles thus
of the previous theorem.
Consider $\zWcal(v_0, e_s)$ in $\Gamma$ and set $\zWcal(v_0, e_s)$ in $\Gamma'$.
It is needed to ensure that no new walks from $v_0$ to the endpoints of $e_s$
appear in $\Gamma'$ that may decrease $t_s(e_s)$ in $\Gamma'$.

Each walk $w'$ in $\zWcal(v_0, e_s)$ of $\Gamma'$ that has no edges of $g'$
is in $\zWcal(v_0, e_s)$ of $\Gamma$.
Let $w'$ be a walk in $\zWcal(v_0, e_s)$ of $\Gamma'$ that has edges of $g'$
in its support $S(w')$.
For walk $w'$, there is a walk $w$ in $\Gamma$ that is not longer than
$w'$.
Walk $w$ contains all edges of $w'$ that do not belong to $g'$;
in addition, for each edge $\langle v_i', v_j' \rangle$ of $w'$ in $t'$,
$w$ contains corresponding edge $\langle v_i, v_j \rangle$ in $w$.
The lengths of such $w'$ and $w$ are equal and, thus, $w'$ and $w$
lie in the same class $[w']_{\approx}$.
For example,
in Figure~\ref{fig:th2_linearize_broom},
consider walk
$w'=\langle v_0, v_1, v_2, v_3, v_4, v_5, v_6, v_7', v_8',$ $v_7',
v_6, v_5, v_4, v_3 \rangle$;
its support is highlighted with gray.
In $\Gamma$, there is a corresponding walk
$w=\langle v_0, v_1, v_2, v_3, v_4, v_7, v_8, v_7, v_4,$ $v_5, v_6, v_5, v_4, v_3
\rangle$.

Thus, for every walk $w'$ in $\Gamma'$, there is a walk in $\Gamma$ that
produces a point on $e_s$ not later than $w'$.
As the result, $e_s$ will be saturated in $\zPcal_\Gamma'$ not earlier than in
$\zPcal_\Gamma$; therefore, $\zPcal_\Gamma'$ belongs to $\mathit{LSTDP}(E)$.

For any vertex $v$ in $\Gamma$ with four or more incident edges,
$v$ may have two odd incident edges only if $v$ lies on a cycle,
$v$ may have three odd incident edges if $v$ is a terminal vertex of $e_s$;
all other incident edges of $v$ are even.
One of even incident edges may belong to the $e_sv$-path;
for each other even incident edge we move it with its connected subgraph to
a bead leaf.
After the procedure is applied repeatedly, we obtain a bead graph with
no vertices of degree four or more.
\end{proof}

As the set of metric graphs that can be constructed from the given finite set of 
edges $E$ is finite, there is a metric graph with the longest stabilization time
built from $E$.
Lemmas~\ref{thm:DP_bead},\ref{thm:DP_beadbroom},\ref{thm:beadthree}
combined give us the following theorem.

\begin{theorem}
Let $E$ be a given set of edges with a length function $l$ on $E$.
There is a $\mathit{DP}$-system on a bead metric graph with vertices of degree three
or less and with an initial point in one of its terminal vertices
that belongs to
$\mathit{LSTDP}(E)$.
\label{coro:fin}
\end{theorem}

As the resulting tree $\mathit{DP}$-system $\zPcal_\Gamma$ in
Lemma~\ref{lemma:slowtree} does not contain cycles at all,
if we fix stabilization edge $e_s$ and apply
the transformations of a graph suggested in the proofs of
Lemmas~\ref{thm:DP_bead},\ref{thm:DP_beadbroom},\ref{thm:beadthree}
to $\zPcal_\Gamma$, then we obtain a linear graph $\zPcal_\Gamma'$.
As such transformations does not add paths of new lengths,
$\zPcal_\Gamma'$ will demonstrate less or equal growth rate than $\zPcal_\Gamma$.
As this can be done for any $\zPcal_\Gamma$, we can improve
Lemma~\ref{lemma:slowtree}.

\begin{corollary}
Let $E$ be a given set of edges with a length function $l$ on $E$,
and $\mathcal{P}_\Gamma$ be a dynamical system with a graph $\Gamma$
constructed from $E$.
Then, there is a dynamical system $\mathcal{P}'_{\Gamma'}$ consisting of
a linear graph $\Gamma'$ constructed from $E$ and one dynamic point
at its vertex,
such that the growth rate of $\mathcal{P}'_{\Gamma'}$ is equal or less than
thus of $\mathcal{P}_\Gamma$.
\end{corollary}

%
%It could be interesting to try to apply it under other restrictions and
%to other models, like timed versions of automata or Petri nets.

It is important to note that the suggested procedures are quite local,
\textit{i.e.,} they are bound to a vertex and stabilization edge.
This means, for example, that we cannot use it, at least straightforwardly,
to construct or prove
that there exists a linear metric graph with many dynamic points
that demonstrates the slowest growth rate at all edges,
as moving subgraphs may reduce the growth rate of one edge but
intensify thus of another.

%%??? check
%%A linear metric graph with edges sorted in the decreasing order started
%%from the vertex with initial point provides the longest linear time,
%%as any path from $e_0$, \textit{i.e.}, edge incident to the initial vertex,
%%back to $e_0$ which involves $e_i$, contains double entries of $e_i$,
%%thus, having long edges close to $e_0$ makes all such paths longer.

%== End of Beads metric graphs with longest stabilization time

%=======================================================================
%== Beads metric graphs with longest saturations time
\section{Saturation time for bead $\mathit{DP}$-systems with incommensurable
edge lengths}

$\mathit{DP}$-system on a metric graph with incommensurable edges
(\textit{i.e.}, the rank of its edges over $\zQbb$ is greater than one)
never stabilizes and, moreover, ergodic
\cite{ChernyshevTolchennikovShafarevich16_BehaviorQuasiParticlesHybridSp_RelGeomGeodesicsAnalyticNT}.
However, it is possible to extend the notion of stabilization time to
systems with incommensurable edges using the notion of $\varepsilon$-net
\cite{hausdorff_set_2005}.

\begin{figure}[!t]
    \begin{center}
        \centerline{\includegraphics[scale=0.9]{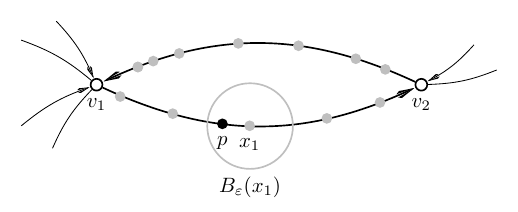}}
        %\centerline{\includegraphics[scale=0.9]{Fig10.pdf}}
        \caption{Edge $\langle v_1,v_2 \rangle$ with
            $\varepsilon$-net $N_{\varepsilon}$ depicted as gray points}
        \label{fig:mg_epsnet}
    \end{center}
\end{figure}

For a point $x$ on $\Gamma$, a closed 1-ball $B_{\varepsilon}(x)$ with
radius $\varepsilon$ around $x$ is the set of all points in $\Gamma$
whose distance from $x$ in the metric $\rho$ is less than $\varepsilon$,
\textit{i.e.,} $\{p\in\Gamma | \rho(p,x)\le \varepsilon \}$.
A finite subset $N_{\varepsilon}$ of points of metric graph $\Gamma$ is
an \emph{$\varepsilon$-net} if, for each point $p\in \Gamma$, distance
$\rho(p,N_{\varepsilon})$ is less than $\varepsilon$.
Thus, the closed 1-balls with centres in $N_{\varepsilon}$ cover
the whole $\Gamma$.
Now we consider \emph{$\varepsilon$-net} over point-places,
and 1-ball $B_{\varepsilon}(x)$ are set of points at $\varepsilon$-distance or
less from moving point-place $x$.
In Figure~\ref{fig:mg_epsnet}, edge $\langle v_1, v_2\rangle$ holds only one
dynamic point $p$, which saturates 1-ball (\textit{i.e.,} line segment bounded
by the circle) $B_{\varepsilon}(x_1)$;
note that, points of $\varepsilon$-net, coloured gray,
are not necessarily regularly distributed.
For $\mathit{DP}$-system $\mathcal{P}_\Gamma$ and $\varepsilon$-net
$N_{\varepsilon}$ on point-places of $\Gamma$,
the saturation time $t_{s}(N_{\varepsilon})$ of $\mathcal{P}_\Gamma$ is
the earliest point in time when, for each point $x$ in $N_{\varepsilon}$,
there is a point $q$ in $\mathcal{P}_\Gamma$ with $\rho(x,q) \le \varepsilon$.
For given $\varepsilon$, the saturation time $t_{s}(\varepsilon)$ of
$\mathcal{P}_\Gamma$ is the supremum of the set of saturation times for all
possible $\varepsilon$-nets on $\Gamma$,
\textit{i.e.}, $t_{s}(\varepsilon) = \sup\{t_{s}(N_{\varepsilon})\}$.
We call 1-ball $B_{\varepsilon}(x)$ on $\Gamma$ with centre $x$ and radius
$\varepsilon$ \emph{saturated} when point $p$ of $\mathcal{P}_\Gamma$
is bound to a place-point in $B_{\varepsilon}(x)$.

\begin{theorem}
Let $E$ be a given set of edges with a length function $l$ on $E$,
and the rank of $E$ over $\zQbb$ is greater than one.
There is a $\mathit{DP}$-system on a bead metric graph with vertices of degree
three or less and with an initial point in one of its terminal vertices
that has the longest saturation time among systems in $\mathit{DP}(E)$.
\label{coro:fin_satur}
\end{theorem}

\begin{proof}
At time $t_{s}(\varepsilon)$, points of $\mathcal{P}_\Gamma$ form an
$\varepsilon$-net themselves; if they do not, it is possible to construct
$N_{\varepsilon}$ that is not saturated at $t_{s}(\varepsilon)$ by taking the centre of a non-covered
ball into $N_{\varepsilon}$, which contradicts the definition of
$t_{s}(\varepsilon)$.
Thus, $t_{s}(\varepsilon)$ is the time point when all the non-saturated balls
vanish.
Fix ball $B_{\varepsilon}(x)$ that is saturated at $t_{s}(\varepsilon)$
by dynamic point $q$ and
build $\varepsilon$-net $N_{\varepsilon}$ such that,
within $B_{\varepsilon}(x)$,
only the centre of $B_{\varepsilon}(x)$ belongs to $N_{\varepsilon}$.
Let $q_0$ in vertex $v_0$ be the ancestor of $q$ in the initial state of
$\mathcal{P}_\Gamma$,
$e_s$ be an edge containing centre $x$, and $p$ be the path of $q$
from $v_0$ to $B_{\varepsilon}(x)$.
By the graph surgery suggested in
Lemma~\ref{thm:DP_bead}--Theorem~\ref{coro:fin},
we build a bead graph with vertex degrees not greater than three that preserves
$v_0e_s$-path $p$.
\end{proof}

%== End of Beads metric graphs with longest stabilization time

\section{Conclusions and future work}
Neither explicit formula nor asymptotic estimates are known for the longest
stabilization time of $\mathit{DP}$-systems constructed from set of edges $E$ or
for the stabilization time of an arbitrary $\mathit{DP}$-system 
in the general case
\cite{ChernyshevTolchennikovShafarevich16_BehaviorQuasiParticlesHybridSp_RelGeomGeodesicsAnalyticNT}.
Theorem~\ref{coro:fin} allows us to obtain the longest stabilization time of
$DP(E)$ by studying only bead graphs of degree not higher than 3.
It also allows to narrow down the search state if
we want to compute the longest stabilization time algorithmically.

%some leads might come from overapproximating it by considering a subclass which
%contains $LSTDP$-systems.
%
%If we want to study the longest stabilization time
%%%For the $DP$-systems constructed from $E$,
%%%Theorem~\ref{coro:fin} allows us to narrow the set to only bead
%%%graphs of degree not higher than 3.
%This, probably, may ease progress towards finding the upper
%bound of stabilization time for an arbitrary $\mathit{DP}$-system,
%which is our distant goal.
%
%%The longest stabilization time is stipulated by a longest path
%%among minimal paths in equivalence classes.
%
%It also allows to narrow down the search state if, for given $E$, 
%we want to find the longest stabilization time of $DP(E)$ algorithmically.
%%%To find the longest stabilization time for the set of metric graphs
%%%constructed from $E$, we may consider only bead graphs.
%The longest stabilization time is stipulated by a longest path
%among minimal paths in equivalence classes.
%%%A linear metric graph with edges sorted in the decreasing order started
%%%from the vertex with initial point provides the longest linear time,
%%%as any path from $e_0$, i.e., edge incident to the initial vertex,
%%%back to $e_0$ which involves $e_i$, contains double entries of $e_i$,
%%%thus, having long edges close to $e_0$ makes all such paths longer.

The proofs are mostly agnostic to the numerical properties of edge
lengths and specific structure of a metric graph.
Thus, it allowed us to extend the results to incommensurable metric graphs
by adapting the notion of stabilization time using $\varepsilon$-nets.
It could be interesting to combine the suggested graph surgery with
\cite{Chernyshev17_SecondTermNumberPointsMetricGraph}

Considering the support set of an $LST$-walk, one could hypothesizes that,
even the absolute difference between longest stabilization time of $DP(E)$
and set of linear graphs built from $E$ could be arbitrary large (Fig.~\ref{fig:MG_uncuttable}),
there is always a linear graph with stabilization time not more than two times
shorter than the longest stabilization time in $DP(E)$.

%\subsubsection*{Acknowledgments.}
%The publication was prepared within the framework of 
%%%the Academic Fund Program at the National Research University 
%%%Higher School of Economics (HSE) in 2017-2018 (grant \textnumero{}17-01-0089)
%%%%\textnumero 
%%%and by the Russian Academic Excellence Project ``5-100''.
 
\subsection*{Acknowledgements}
%\acknowledgements
The reported study was funded by RFBR, project number 20-07-01103a.
Author expresses gratitude to V.~Chernyshev for providing the state-of-the-art
view of the field and pointing to some recent results, and to A.~G.~Khovanskii
for explaining the useful notion of $\varepsilon$-net.

\bibliography{leo_norus}
%\printbibliography

\end{document}